\documentclass[12pt,reqno]{amsart}

\usepackage{amsmath,amssymb,amsthm,bm,mathrsfs,braket,sasaki,color}
\usepackage{graphicx}
\makeatletter
\@addtoreset{equation}{section}
\makeatother

\theoremstyle{plain}
\newtheorem{theorem}{Theorem}[section]
\newtheorem{proposition}[theorem]{Proposition}
\newtheorem{lemma}[theorem]{Lemma}

\theoremstyle{definition}

\theoremstyle{remark}

\title{One Particle Binding of Many-Particle Semi-Relativistic Pauli-Fierz Model}
\author{Itaru Sasaki}
\thanks{Email address: \texttt{isasaki@shinshu-u.ac.jp} \\
\indent This work was partly supported by Research supported by KAKENHI Y22740087, 
and was performed through the Program for Dissemination of Tenure-Track System  
funded by the Ministry of Education and Science, Japan}
\address{ Fiber-Nanotech Young Researcher Empowerment Center, 
 Shinshu University, Matsumoto 390--8621, Japan}
\email{isasaki@shinshu-u.ac.jp}

\keywords{binding condition; relativistic particle; quantum electrodynamics; functional integration}
\subjclass[2000]{35P05, 35P15}

\begin{document}
%%%%%%%%%%%%%%%%
\begin{abstract}
It is shown that at least one particle is bound in the $N$-particle semi-relativistic 
Pauli-Fierz model with negative potential $V(\bx)$.
It is assumed that the particles have no spin and obey the Bose or Boltzmann statistics,
and the one particle Hamiltonian $\sqrt{-\Delta+M^2}-M+ V(\bx)$ has 
a ground state with negative energy $-e_0<0$, where $M>0$ denotes the mass of the particle. 
We show that the ground state energy of the total system $E^V(N)$ is less than $E^0(N)-e_0$.
% where $E^0(N)$ is the ground state energy of the total Hamiltonian without $V(\bx)$.
 \end{abstract}
%%%%%%%%%%%%%%%%

%%%%%%%%%%%%%%%%
\maketitle

%%%%%%%%%%%%%%%%
\section{Introduction}
We consider a quantum system of $N$-charged relativistic particles interacting with the quantum electromagnetic field
and the fixed nuclear potential $V(\bx)\leq 0$.
The Hamiltonian of this system is defined by
\begin{align}
 H^V := \sum_{i=1}^N T_\bA(\bp_i) + \sum_{i=1}^NV(\bx_i) + H_f + \sum_{i<j}W(\bx_i-\bx_j), \label{hamil}
\end{align}
where $\bx_i$ denotes the position of the $i$-th particle and  $T_\bA(\bp_i)$ denotes the kinetic energy
of the $i$-th particle which depends on the momentum $\bp_i:=-i\nabla_{\bx_i}$ and the quantized electromagnetic potential $\bA(\bx_i)$.
$H_f$ denotes the free photon Hamiltonian and $W$ is the interparticle potential energy. 
In this paper, the $N$-particles are treated as relativistic particles and we take the relativistic kinetic energy
\begin{align}
 T_\bA(\bp_i) = \sqrt{(\bp_i-q\bA(\bx_i)) ^2 + M^2}-M,   \label{relkin}
\end{align}
where $q$ and $M$ denote the charge and mass of the particle, respectively. 
The system described by the Hamiltonian \eqref{hamil} is called the semi-relativistic Pauli-Fierz model.
We assume that the particles have no spin and obey the Bose-Einstein or Boltzmann statistics.
We are interested in whether the nuclear potential $V$ can bind the particle.

Let $E^V(N)$ be the lowest energy of $H^V(N)$. Note that $E^0(N)$ is $E^V(N)$ with $V=0$. 
In order to prove the existence of ground state for $H^V$, it is very important to show 
the inequality
\begin{align}
  E^V(N) < \min \{ E^V(N-N')+E^0(N')| N'=1,2,\dots,N \}. \label{bindcond}
\end{align}
This inequality is called the binding condition.
In this paper, we prove that 
\begin{align}
 E^V(N) \leq E^0(N) - e_0,   \label{mainineq}
\end{align}
where $-e_0$ is the ground state energy of the one particle Hamiltonian
\begin{align}
  h^V = \sqrt{-\Delta+M^2}-M +V(\bx).
\end{align}
We assume that $-e_0<0$. Then \eqref{mainineq} implies the strict inequality $E^V(N)<E^0(N)$,
which is weaker case of \eqref{bindcond}.
Physically, this inequality means that at least one particle is bound in the lowest energy state.
When $N=1$, the inequality \eqref{bindcond} becomes $E^V(1)<E^0(1)$.
The inequality $E^V(1)\leq E^0(1)- e_\mathrm{NR,0}$ was proved by \cite{KMS:2011-1}
(see also \cite{2011RvMaP..23..375K}),
 where $-e_\mathrm{NR,0}$ is the ground state energy of the non-relativistic 
particle Hamiltonian $-(1/2M)\Delta+V$.
The inequality $E^V\leq E^0 -e_0$ with the relativistic ground state energy 
$-e_0:=\inf\spec(\sqrt{-\Delta + M^2}-M+V(\bx))$ was shown in \cite{HS10}. 
More better bound including the effect of the mass renormalization was given in \cite{2012arXiv1207.5638K}.
The spectrum of the polaron of this relativistic model was studied by \cite{MR2498761}.

In the case where the particles have non-relativistic kinetic energy
\begin{align}
    T_\bA(\bp_i) = \frac{1}{2M}(\bp_i-q\bA(\bx_i))^2,
\end{align}
the system is called the Pauli-Fierz model, and this model was widely studied.
The most important result on the existence of ground state is the paper \cite{Griesemer-Lieb-Loss:2001},
where, for the non-relativistic case, it is proved that the inequality \eqref{bindcond} implies
 the existence of ground state. 
The binding condition for an atomic Coulomb system was proved in the continuous paper \cite{Lieb-Loss:2003}.
The non-relativistic version of \eqref{mainineq} was originally shown in \cite{Griesemer-Lieb-Loss:2001}.
Our result \eqref{mainineq} can be considered as a relativistic improvement of \cite[Theorem 3.1]{Griesemer-Lieb-Loss:2001}.
We have used the method they had developed in \cite{Griesemer-Lieb-Loss:2001} with some modification.
But, unfortunately, our method can not be applied for the Fermionic or spinor particles.

The difficulties to prove \eqref{mainineq} come from the relativistic kinetic energy \eqref{relkin} which is 
clearly non-local.
The key idea of our proof is to use the convexity of the kinetic energy $T_\bA(\bp_i)$ when estimating 
 the energy expectation of a test function. 
The convexity of the kinetic energy follows from the property of the semi-group of the Hamiltonian
which is positivity preserving(\cite{LHB11}).
The semi-group is, however, positivity preserving only for the case when the particles are spinless and 
obey the Bose-Einstein or Boltzmann statistics.
But, it is remarkable that although the relativistic Schr\"odinger operator with a classical magnetic vector potential $\bA(\bx)$ 
may not have the convexity, relativistic kinetic energy $T_\bA(\bp_i)$ is convex.

In Sect. 2, we give the rigorous definition of the system, and state the main result. 
In Sect. 3, we give the proof of the main theorem. 

\section{Definition and Main Result}
The Hilbert space for the $N$-particle state is defined by 
\begin{align}
  \cH_\mathrm{part} = L^2(\BR^{3N}). \label{21fde}
\end{align}
When the $N$ particles obey the Bose-Einstein statistics, one needs to take $\otimes_\mathrm{sym}^N L^2(\BR^3)$ instead of \eqref{21fde} where $\otimes_\mathrm{sym}$ denotes the symmetric tensor product.
Almost all the discussions in this paper are independent of the such choice of statistics.
Hence we only consider the case of \eqref{21fde}.

The position of the particles are denoted by $\underline{\bX}=(\bx_1,\cdots,\bx_N)\in\BR^{3N}$
with $\bx_i=(x_i^1,x_i^2,x_i^3)\in \BR^3$, $i=1,\cdots,N$.
The Hilbert space for the photon field is the Fock space
\begin{align}
  \cH_\mathrm{phot} := \bigoplus_{n=0}^\infty \left[ \Tensor_\mathrm{sym}^n L^2(\BR^3\times \{1,2\})\right],
\end{align}
with $\tensor_\mathrm{sym}^0 L^2(\BR^3\times\{1,2\}) =:\BC$.
The Hilbert space for the semi-relativistic Pauli-Fierz model is defined by
\begin{align}
    \cH :=  \cH_\mathrm{part} \tensor \cH_\mathrm{phot},
\end{align}
The smeared creation and annihilation operators in $\cH_\mathrm{phot}$ are denoted by $a(f)^*, a(f), f\in L^2(\BR^3\times \{1,2\})$, respectively.
The vacuum vector is defined by $\Omega_\mathrm{phot}:=1\oplus 0\oplus 0 \cdots \cH_\mathrm{phot}$.
For a closed operator $T$ on $L^2(\BR^3\times \{1,2\})$, the second quantization of $T$ is denoted by 
$d\Gamma(T) :\cH_\mathrm{phot}\to \cH_\mathrm{phot}$. Let $\ome:\BR^3\to [0,\infty)$ be a Borel measurable function such that $0<\ome(\bk)<\infty$.
We also denote by the same symbol $\ome$ the multiplication operator by the function $\ome$, which acts in $L^2(\BR^3\times\{1,2\})$
as $(\ome f)(\bk,\lambda) = \ome(\bk)f(\bk,\lambda)$, $(\bk,\lambda)\in \BR^3\times\{1,2\}$.
The free Hamiltonian of the photon field is defined by
\begin{align}
  H_f := d\Gamma(\ome)
\end{align}
Let $\be^{(\lambda)}:\BR^3\to \BR^3, \lambda=1,2$ be polarization vectors, which is defined by
\begin{align}
  \be^{(\lambda)}(\bk) \cdot \be^{(\mu)}(\bk) = \delta_{\lambda,\mu}, \quad 
  \bk \cdot \be^{(\lambda)}(\bk) = 0, \qquad \bk \in\BR^3, ~ \lambda,\mu\in\{1,2\}. 
\end{align}
We write $\be^{(\lambda)}(\bk)= (e_1^{(\lambda)}(\bk),e_2^{(\lambda)}(\bk),e_3^{(\lambda)}(\bk))$ and 
suppose that each component $e_j^{(\lambda)}(\bk)$ is a Borel measurable function in $\bk$.
Let $\Lambda \in L^2(\BR^3)$ be a function such that 
\begin{align}
  \ome^{-1/2} \Lambda \in L^2(\BR^3).
\end{align}
For $j=1,2,3$, we set 
\begin{align}
   g_j(\bk,\lambda;\bx) := \ome(\bk)^{-1/2} \Lambda(\bk) e_j^{(\lambda)}(\bk) e^{-i\bk\cdot\bx},
   \qquad (\bk,\lambda) \in \BR^3\times\{1,2\}, ~ \bx\in\BR^3.
\end{align}
For each $\bx\in\BR^3$, $g_j(\bx)=g_j(\cdot,\cdot;\bx)$ can be regarded as an element of $L^2(\BR^3\times \{1,2\})$.
Then, the quantized electromagnetic field at $\bx\in\BR^3$ is defined by
\begin{align}
  A_j(\bx) := \frac{1}{\sqrt{2}} \overline{[a(g_j(\bx))+a^*(g_j(\bx))]},
\end{align}
where $\bar{T}$ denotes the closure of closable operator $T$.
The quantized electromagnetic field $\bA(\bx):=(A_1(\bx),A_2(\bx),A_3(\bx))$ satisfies the Coulomb gauge condition:
\begin{align}
  \sum_{j=1}^3 \frac{\del A_j(\bx)}{\del x^j} = 0, \qquad \bx =(x^1,x^2,x^3).
\end{align}
The Hilbert space $\cH$ can be identified as 
\begin{align}
  \cH \cong  \int_{\BR^{3N}}^\oplus \cH_\mathrm{phot} d^{3N}\underline{X}, \qquad \underline{X}=(\bx_1,\cdots, \bx_N) \in \BR^{3N}.
\end{align}
The quantized electromagnetic field on the total Hilbert space is defined by the fiber direct integral of $A_j(\bx)$:
\begin{align}
  A_j(\hat{\bx}_i) := \int_{\BR^{3N}}^\oplus A_j(\bx_i) d\underline{X}.
\end{align}

Let $C_\mathrm{c}^\infty$ be the set of infinitely differentiable functions with compact support. Let 
\begin{align}
 \cF_\mathrm{fin} := \cL[ \{ a^*(f_1)\cdots a^*(f_n)\Omega_\mathrm{phot}, \Omega_\mathrm{phot} 
                                            | f_j \in C_\mathrm{c}^\infty(\BR^3\times\{1,2\}), j=1,\cdots,n, \, n\in \BN \}]
\end{align}
be a finite photon subspace spanned by $C_\mathrm{c}^\infty(\BR^3\times\{1,2\})$. The subspace
\begin{align}
     \cD := (\hat\tensor_\mathrm{sym}^N C_\mathrm{c}^\infty(\BR^{3})) \hat\tensor \cF_\mathrm{fin}
\end{align}
is dense in $\cH$, where $\hat\tensor $ denotes the algebraic tensor product.
In what follows for notational convenience we omit the symbol $\otimes$ in $L^2(\BR^{3N})\otimes \cH_\mathrm{phot}$.
For two sets of operators $\mathbf{a}=(a_1,a_2,a_3), \mathbf{b}=(b_1,b_2,b_3)$, we denote
$\inner{\mathbf{a}f}{\mathbf{b}g}$ by $\sum_{j=1}^3 \inner{a_jf}{b_jf}$.
We define the non-negative quadratic form on $\cD\times \cD$ by 
\begin{align}
   K_{i,\bA}(\Psi,\Phi) =   \inner{(\bp_i - q \bA(\hat{\bx}_i))\Psi}{(\bp_i - q\bA(\hat{\bx}_i))\Phi}
                                      +M^2 \inner{\Psi}{\Phi},
\end{align}
for $i=1,\dots,N$, where $\bp_i := -i\nabla_{\bx_i} = -i ( \del_{x_i^1}, \del_{x_i^2}, \del_{x_i^3})$.
Note that $K_{i,\bA}$ is a positive closable form, and we denote its closure by $\bar{K}_{i,\bA}$.
Let $L_{i,\bA}$ be the self-adjoint operator associated with $\bar{K}_{i,\bA}$, i.e.,
\begin{align}
    \dom (L_{i,\bA}^{1/2}) &= Q(\bar{K}_{i,\bA}), \\
    \bar{K}_{i,\bA}(\Psi,\Phi) &= \inner{L_{i,\bA}^{1/2}\Psi}{L_{i,\bA}^{1/2}\Phi},
\end{align}
for all $\Psi,\Phi \in Q(\bar{K}_{i,\bA})$. Since $\cD\subset Q(L_{i,\bA})$, we have $\cD \subset \dom(L_{i,\bA}^{1/2})$.
We set
\begin{align}
   \sqrt{(\bp_i - q \bA(\hat\bx_i))^2 +M^2} - M :=  T_{\bA}(\bp_i)  := L_{i,\bA}^{1/2} - M.
\end{align}
The Hamiltonian of the $N$-particle semi-relativistic Pauli-Fierz model is defined by
\begin{align}
 H^V := \sum_{i=1}^N\left( \sqrt{ (\bp_i - q \bA(\hat\bx_i)^2 +M^2}-M+V(\bx_i) \right)
       + H_f + \sum_{1\leq i<j\leq N}W(\bx_i - \bx_j),
\end{align}
where $V:\BR^3\to \BR$ and $W(\bx):\BR^3\to \BR$ are measurable functions.
We set $H^0 := H^V|_{V=0}$.
We introduce the following conditions:
\begin{enumerate}
 \item[(H.1)] $\ome^{3/2} \Lambda \in L^2(\BR^3)$.
 \item[(H.2)]
$V(\bx)$ and $W(\bx)$ are relatively compact with respect to the three dimensional 
	      relativistic Schr\"odinger operator$\sqrt{-\Delta_\bx+1}-1$ and the relative bounds are strictly smaller than one.
 \item[(H.3)] The self-adjoint operator $h^V := \sqrt{-\Delta+M^2}-M+V(\bx)$  has a negative energy ground state $-e_0<0$.
 \item[(H.4)] $V(\bx)\leq 0$ for all $\bx\in\BR^3$.
\end{enumerate}
The essential self-adjointness was proved in \cite[Corollary 7.60]{LHB11}.

%%%%%%%%%%%%%%%%
\begin{proposition}(Essential self-adjointness)
Assume (H.1) and (H.2). 
Then,  the Hamiltonians $H^V, H^0$ are essentially self-adjoint on $\cD$.
\end{proposition}
%%%%%%%%%%%%%%%%

We denote the closure of $H^V$ and $H^0$ by the same symbol.
Let $E^0(N) = \inf\spec(H^0)$ and $E^V(N)=\inf\spec(H^V)$ are the ground state energies.
The main result in this paper is the following:
%%%%%%%%%%%%%%%%
\begin{theorem}{\label{thm}}
 Assume (H.1)--(H.4). Then,  for all $q\in\BR$ and $M\geq 0$, the inequality
\begin{align}
   E^V(N) \leq E^0(N) - e_0  \label{thmineq}
\end{align}
holds.
\end{theorem}
%%%%%%%%%%%%%%%%

\section{Proof of Theorem \ref{thm}}
We start from the the following basic fact.
 \begin{lemma}{\label{lem1}}
 Let $(Q,\Sigma,\mu)$ be a $\sigma$-finite measure space,
and $T$ be a positivity preserving bounded operator on $L^2(Q,d\mu)$.
Then, for all non-negative $f,g\in L^2(Q,d\mu)$, the following holds
\begin{align}
  (Tf)(q)^2 + (Tg)(q)^2 \leq [(T(f^2+g^2)^{1/2})(q)]^2, \quad \mu\text{-a.e.}~q\in Q.
\end{align}
\end{lemma}
\begin{proof}
First we assume that $f,g$ are non-negative simple functions, i.e.,
\begin{align}
  f(q) = \sum_{i=1}^n \alpha_i \chi_{A_i}(q), \quad 
  g(q) = \sum_{i=1}^n \beta_i \chi_{A_i}(q), 
\end{align}
with $\alpha_i,\beta_i \geq 0$ and $A_i \in\Sigma$,
where $\chi_{A_i}$ is the characteristic function of $A_i$.
 One can assume that $A_i\cap A_j = \emptyset$.
Then, by noting that $T\chi_{A_i}\geq 0$, we have
\begin{align}
 (Tf)^2+ (Tg)^2
 &= \sum_{i}\sum_{j} (\alpha_i\alpha_j + \beta_i \beta_j) (T\chi_{A_i})(T\chi_{A_j}) \\
 &\leq \sum_i\sum_{j} (\alpha_i^2+\beta_i^2)^{1/2} (\alpha_j^2+\beta_j^2)^{1/2} (T\chi_{A_i})(T\chi_{A_j}) \\
 &=  \left( T\sum_{i}(\alpha_i^2+\beta_i^2)^{1/2}\chi_{A_i} \right)^2 \\
 &= \left( T (\sum_{i} \alpha_i^2\chi_{A_i} + \sum_i \beta_i^2 \chi_{A_i})^{1/2} \right)^2\\
&  = (T(f^2+g^2)^{1/2})^2 \label{lemineq1}
\end{align}
For any $f,g \in L^2(Q,d\mu)$, there exist simple functions $f_n,~g_n$ such that 
$0\leq f_n \leq f$, $0\leq g_n \leq g$ and $f_n(q)\nearrow f(q)$, $g_n(q) \nearrow g(q)$, $\mu$-a.e.$q$ as $n\to\infty$.
By \eqref{lemineq1}, we have
\begin{align}
  (Tf_n)(q)^2 + (Tg_n)(q)^2 \leq [(T(f_n^2+g_n^2)^{1/2})(q)]^2 \leq [(T(f^2+g^2)^{1/2})(q)]^2
\end{align}
for $\mu$-a.e. $q\in Q$. 
Since $T$ is bounded, we have $\norm{T(f-f_n)}\to 0$ as $n\to\infty$.
By taking the subsequence $\{n_j\}_j$, we have 
\begin{align}
  \lim_{j\to\infty}   (Tf_{n_j})(q)^2 + (Tg_{n_j})(q)^2 
  = (Tf)(q)^2 + (Tg)(q)^2 \leq [ (T(f^2+g^2)^{1/2})(q) ]^2,
\end{align}
for $\mu$-a.e.$q$.
\end{proof}
%%%%%%%%%%%%%%%%
For a semi-bounded self-adjoint operator $h$, we denote the associated quadratic form 
by $(f,hg)$, $f,g\in Q(h)$.
As a consequence of Lemma \ref{lem1}, we have the following fact:
%%%%%%%%%%%%%%%%
\begin{lemma}{\label{lem2}}
 Let $h$ be a semi-bounded self-adjoint operator on a $L^2$-space such that
$e^{-th}$ is positivity preserving for all $t>0$.
Then, for all $f\in \dom(h)$, $|f| \in Q(h)$ and
\begin{align}
  (|f|,h|f|) \leq \inner{f}{hf}. \label{310}
\end{align}
In particular, for non-negative $f,g \in Q(h)$, $\sqrt{f^2+g^2} \in Q(h)$ and
\begin{align}
  \left(\sqrt{f^2+g^2},h\sqrt{f^2+g^2}\right) \leq (f,hf) + (g,hg)    \label{311}
\end{align}
holds.
\end{lemma}
%%%%%%%%%%%%%%%%
\begin{proof}
 Note that $u\in Q(h)$ if and only if $t^{-1}\inner{u}{(1-e^{-th})u}$ converges as $t\to 0$.
Assume that $f\in \dom(h)$. Then
\begin{align}
  \inner{f}{hf} = \lim_{t\to 0} t^{-1}(f, (1-e^{-th})f) \geq \lim_{t\to 0} t^{-1}(|f|,(1-e^{-th})|f|)
  = (|f|,h|f|) >-\infty,
\end{align}
which proves \eqref{310}. Next we assume $f,g\in Q(h)$. By Lemma \ref{lem1}, we have
\begin{align}
 (f,hf)+(g,hg) &= \lim_{t\to 0} t^{-1}\left[ (f,(1-e^{-th})f) + (g, (1-e^{-th})g) \right] \\
                    &\geq \lim_{t\to 0} t^{-1}  (\sqrt{f^2+g^2},(1-e^{-th}) \sqrt{f^2+g^2})\\
                   & = (\sqrt{f^2+g^2}, h \sqrt{f^2+g^2}) > -\infty,
\end{align}
which implies that $\sqrt{f^2+g^2} \in Q(h)$ and \eqref{311} holds.
\end{proof}

%%%%%%%%%%%%%%%%
\begin{proof}[Proof of Theorem  \ref{thm}]
By using the functional integration, it is proved that 
there exists a $\sigma$-finite measure space $(\cL_\mathrm{E}, \Sigma_E, \mu_E)$
and unitary operator $U:\cH_\mathrm{phot} \to L^2(\cL_\mathrm{E},d\mu_E)$ such that 
$I\tensor U e^{-tH^0} I\tensor U^{-1}$ is positivity preserving
(see \cite[Corollary 7.64]{LHB11}).
We set $\tilde{H}^0 = I\tensor U H^0 I \tensor U^{-1}$.
For arbitrary fixed $\ep>0$, we can choose normalized vectors $F \in \cD$ and
 $\phi \in C_\mathrm{c}^\infty(\BR^3)$ such that
\begin{align}
& \inner{F}{H^0F} < E^0(N) + \ep \\
& \inner{\phi}{h^V\phi} < -e_0+\ep \\
& \phi(\bx) \geq 0, \quad \bx\in\BR^3.
\end{align}
For each $\by\in\BR^3$, we define the translation operator
\begin{align}
  \cT_\by := \exp\left(-i\by\cdot\sum_{i=1}^N \bp_i \right) \tensor \exp(-i\by\cdot d\Gamma(\bk)).
\end{align}
One can show that $\cT_\by \cD = \cD$ and $H^0$ is translation invariant, 
i.e., $\cT_\by^{-1} H^0 \cT_\by= H^0$.
We set $\til{F}=(I\tensor U) F \in L^2(\cL_\mathrm{E},d\mu_\mathrm{E})$. Our test function is 
\begin{align}
 \Phi_\by 
 = \left[\sum_{i=1}^N \phi(\hat\bx_i)^2\right]^{1/2} \cT_\by I\tensor U^{-1} |\til{F}|,
\end{align}
where $\phi(\hat\bx_i)$ denotes the multiplication operator by the function $\phi(\bx_i)$.
Note that
\begin{align}
\int_{\BR^3}d\by \norm{\Phi_\by}^2 
= \sum_{i=1}^N \norm{\phi(\bx_j)}_{L^2(\BR^3)}^2 \cdot \norm{|\til{F}|}^2 =N,
\end{align}
and
\begin{align}
& \int_{\BR^3} d\by \inner{\Phi_\by}{\sum_{i=1}^N V(\bx_i)\Phi_\by} \\
& = \sum_{i,j}\int_{\BR^3} d\by \phi(\bx_i+\by)^2 V(\bx_j+\by) \inner{F(\bX)}{F(\bX)}_{\cH_\mathrm{phot}} \\
& \leq \sum_{i=1}^N \int_{\BR^3} d\by \phi(\bx_i+\by)^2 V(\bx_i+\by) \inner{F(\bX)}{F(\bX)}_{\cH_\mathrm{phot}} \\
& =N \inner{\phi}{V\phi}, \label{325}
\end{align}
where we used the condition (H.4).
By Lemma \ref{lem2}, we have $\Phi_\by \in Q(H^0)$ and
\begin{align}
 (\Phi_\by, H^0\Phi_\by)
& = \left(\left[\sum_{i=1}^N \phi(\hat\bx_i+\by)^2\right]^{1/2}|\til F|, \til{H}^0\left[\sum_{i=1}^N \phi(\hat\bx_i+\by)^2\right]^{1/2}|\til{F}|\right) \\
&= \left(\left[\sum_{i=1}^N \phi(\hat\bx_i+\by)^2|\til{F}|^2\right]^{1/2}, \til{H}^0\left[\sum_{i=1}^N \phi(\hat\bx_i+\by)^2 |\til{F}|^2\right]^{1/2}\right) \\
& \leq  \sum_{i=1}^N \left( \phi(\hat\bx_i+\by)|\til{F}|, \til{H}^0 \phi(\hat\bx_i+\by)|\til{F}| \right) \\
& \leq  \sum_{i=1}^N \left( \phi(\hat\bx_i+\by)\til{F}, \til{H}^0 \phi(\hat\bx_i+\by) \til{F} \right) \\
& =  \sum_{i=1}^N \inner{\phi(\hat\bx_i+\by) F}{ H^0 \phi(\hat\bx_i+\by) F } \label{330}
\end{align}
%%%%%%%%%%%%%%%%
 \begin{lemma}{\label{lem3}} For $i=1,\dots,N$, we have
\begin{align}
& \int_{\BR^3}d\by \inner{\phi(\hat\bx_i+\by)F}{H^0 \phi(\hat\bx_i+\by)F} \\
& \leq \inner{F}{H^0 F} + \inner{\phi}{(\sqrt{-\Delta+M^2}-M)\phi}_{L^2(\BR^3)}.
 \label{332}
\end{align}
 \end{lemma}
%%%%%%%%%%%%%%%%
 \begin{proof}
The proof of the lemma is essentially same as the proof of \cite[Corollary 3.3]{HS10}.
So we omit it.
 \end{proof}
By combining estimates \eqref{325}, \eqref{330} and \eqref{332}, we have
\begin{align}
 \int_{\BR^3}d\by (\Phi_\by, H^V\Phi_\by)
 &\leq N \inner{F}{H^0F} + N \inner{\phi}{h^V\phi} \\
 & \leq N (E^0(N)-e_0+2\ep),
\end{align}
which implies that there exist $\by\in\BR^3$ such that $\norm{\Phi_\by}\neq 0$ and 
\begin{align}
  E^V(N)\norm{\Phi_\by}^2 
 \leq  (\Phi_\by, H^V\Phi_\by) 
 < (E^0(N) - e_0+2\ep) \norm{\Phi_\by}^2.
\end{align}
Since $\ep>0$ is arbitrarily, inequality \eqref{thmineq} holds.
 \end{proof}
%%%%%%%%%%%%%%%%

\section*{Acknowledgments}
I. S. thanks F. Hiroshima for his comments on the theory of functional integration.

%%コメントアウト
%\bibliographystyle{amsplain}
%\bibliography{sasaki-biblio}

%%%元に戻すには次から下を削除して上のコメントアウトを復活
\providecommand{\bysame}{\leavevmode\hbox to3em{\hrulefill}\thinspace}
\providecommand{\MR}{\relax\ifhmode\unskip\space\fi MR }
\providecommand{\MRhref}[2]{%
  \href{http://www.ams.org/mathscinet-getitem?mr=#1}{#2}
}
\providecommand{\href}[2]{#2}

\end{document}